\documentclass{article}[12pt]
  
\usepackage{latexsym,color}
\usepackage{graphicx}
\usepackage{mathptmx}
\usepackage{amssymb}
\usepackage{amsfonts}
\usepackage{amsmath}
\usepackage{latexsym}

\numberwithin{equation}{section}
\newtheorem{thm}{Theorem}[section]
\newtheorem{theorem}{Theorem}[section]
\newtheorem{definition}[thm]{Definition}
\newtheorem{lemma}[thm]{Lemma}
\newtheorem{remark}[thm]{Remark}
\newtheorem{corollary}[thm]{Corollary}
\newtheorem{proposition}[thm]{Proposition}
\newtheorem{asm}[thm]{Assumption}

\newcommand{\cF}{\mathcal{F}}

\newcommand{\cL}{\mathcal{L}}
\newcommand{\cM}{\mathcal{M}}

\newcommand{\cP}{\mathcal{P}}
\newcommand{\cQ}{\mathcal{Q}}

\newcommand{\cW}{\mathcal{W}}

\newcommand{\Om}{{\Omega}}
\newcommand{\om}{{\omega}}

\newcommand \E{\mathbb{E}}
\newcommand \F{\mathbb{F}}

\newcommand \Q{\mathbb{Q}}
\newcommand \R{\mathbb{R}}

\newcommand \Sb{\mathbb {S}}

\def\reff#1{{\rm(\ref{#1})}}

\newenvironment{proof}{\noindent {\bf Proof.\/}}{$\qed$\vskip 0.1in}
\def\qed{ \hfill \vrule width.2cm height.2cm depth0cm\smallskip}

\begin{document}

\title{Robust Hedging 
with Proportional Transaction Costs
\thanks{Research partly supported by the
European Research Council under the grant 228053-FiRM,
 by the ETH Foundation and by the Swiss Finance Institute.
 Authors would like to thank Lev Buhovsky, Jan Ob{\l}{\'o}j
 and Josef Teichmann  for
 insightful discussions and comments.}
}

\author{
  Yan Dolinsky
  \thanks{
  Hebrew University, Dept.\ of Statistics,
  \texttt{yan.dolinsky@mail.huji.ac.il}
  }
  \and
  H.\ Mete Soner
  \thanks{
  ETH Zurich, Dept.\ of Mathematics,
  and Swiss Finance Institute, \texttt{hmsoner@ethz.ch}}
  }
  
\date{\today}

\maketitle

\begin{abstract}
Duality  for robust hedging with
proportional transaction costs
of path dependent European options
is obtained in a
discrete time financial market with one risky asset.
Investor's portfolio consists of a dynamically traded stock and a static
position in vanilla options which can be exercised at maturity.
Trading  of both the options and the stock
are subject to proportional
transaction costs. The main theorem is
duality between hedging
and a Monge-Kantorovich type
optimization problem.
In this dual transport problem
the optimization is over all
the probability measures which
satisfy an approximate martingale condition
related to consistent price systems
in addition to
an approximate marginal constraints.
\end{abstract}
\vspace{6pt}

{\small
\vspace{4pt}

\vspace{4pt}


\section{Introduction}
\label{sec:1}

As well known super-replication
in markets with transaction costs is quite costly
\cite{ssc,ls}.
Naturally the same is
even more true for the model free case
in which one does not place any probabilistic
assumptions on the behavior of the risky asset.
However, one may reduce
the hedging cost
by including  liquid derivatives in the super-replicating
portfolio. In particular, one may use
all call options (written on the
underlying asset)
with a price that is
known to the investor initially.
This leads us to the semi--static hedging
introduced in the classical paper of Hobson
\cite{H} in markets without transaction costs.
So following \cite{H},
we assume that all call options are traded assets
and can be initially bought or sold
for a known price.
In addition to these static option positions,
the stock is also traded dynamically.
These trades, however, are subject to
transaction costs.  Each option has
its own cost and
their general structure
is outlined in the next section.

In this market, we consider the problem of
robust hedging of
a given path dependent European option.
Robust hedging
refers to super-replication
of an option
for all possible stock price processes.
This approach
has been actively researched over the past decade
since the seminal paper of Hobson \cite{H}.
In particular, the optimal portfolio is
explicitly constructed for special cases of European options in continuous time;
barrier options in \cite{BHR} and \cite{CO,CO1},
lookback options in \cite{GLT},  \cite{LT} and  \cite{H}, and
volatility options in \cite{CL}.
The main technique that is employed in these papers
is the Skorohod embedding.
For more information,
we refer the reader to the
surveys of Hobson  \cite{H1},
Ob{\l}{\'o}j \cite{O}
and to the reference therein.

Recently, an alternative approach
is developed using
the connection  to optimal
transport.  Duality
results in different types of generality
or modeling have been proved
in \cite{ABS}, \cite{BHLP}, \cite{DS} and \cite{GLT}
in frictionless markets.
In particular, \cite{DS} studies the continuous
time models, \cite{GLT} provides the connection
to stochastic optimal control and a general
solution methodology,
\cite{BHLP} proves a general duality in discrete time
and \cite{ABS} studies the question of
fundamental theorem of asset pricing
in this context.

Although much has been established,
the effect of frictions -  in particular the impact of
transaction costs - in this context is not fully
studied.
The classical probabilistic models
with transaction costs, however, is well studied.
In the classical model,
a stock price model is assumed
and  hedging is done only
through the stock and no static position
in the options is used.
Then, the dual
is given as the supremum
of  ``approximate" martingale measures
which are equivalent to the market probability measure,
see  \cite{Sc,KS} and the references therein.
In this paper, we extend this result to the robust case.
Namely, we prove that the super--replication price can be represented
as a martingale optimal transport problem.
The dual control problem is the supremum of the
expectation of the option, over all approximate martingale measures
which also satisfy an approximate marginal condition
at maturity.
This result is stated in Theorem \ref{thm2.1} below
and the definition of an approximate martingale
is given in Definition \ref{d.app}.
Indeed, approximate martingales
are very closely related
to consistent price systems which play
a central role in the duality theory for
markets with proportional transaction costs.

Recently
Acciaio, Beiglb\"ock and Schachermayer \cite{ABS}
proved the fundamental theorem
of asset pricing (FTAP) in discrete time markets
without transaction costs in the robust setting.
Also Bouchard and Nutz \cite{BN} studies
FTAP again in discrete time but in
the quasi-sure setting.
Our main duality result  has
implications towards these results as well.
These corollaries are discussed
in the subsection \ref{ss.ftap}.

As in our previous paper \cite{DS} on robust hedging,
our proof relies on discretization of the problem.
We first show that
the original robust hedging problem
can be obtained as a limit of hedging
problems
that are defined on {\em finite spaces}.
We exploit  the finiteness of
these approximate problems
and directly  apply an elementary
Kuhn--Tucker duality theory.
We then prove that any sequence
of probability measures that
are asymptotical
maximizers of these finite problems
is  tight.
The final step is then
to directly use
weak convergence and pass to the limit.

The paper is organized as follows.
Main results are formulated in the next section
and proved in Section 3.
The final section is devoted to
the proof of an auxiliary result that is used in the
proof of the main results. This auxiliary result deals with
super--replication under constraints and
maybe of independent interest.

\section{Preliminaries and main results}
\label{sec:2}
\setcounter{equation}{0}

The  financial market consists of
 a savings account $B$ and
 a risky asset $S$
 and the trading is restricted to
 finitely many time points.
 Hence, the stock price
 process is
 $S_k$, $k=0,1,...,N$,
 where $N<\infty$ is
the maturity date or the total number of allowed
trades.
 By discounting,
 we normalize $B\equiv 1$.
 Furthermore, we normalize the
 initial stock price  $s:=S_0>0$
 to one as well.
 Then, the set $\Omega$ of
 all possible price processes
 is simply all vectors
 $(\om_0,\ldots,\om_N)\in \mathbb{R}_{+}^{N+1}$
which satisfy $\om_0=1$ and $\om_1,\ldots,\om_N\geq0$.
Then, any element of $\Om$ is a possible
stock price process.
So we let $\Sb$ be the canonical
process given by $\Sb_k(\om):= \om_k$
for $k=0,\ldots,N$.
 Let us emphasize
 we make no assumptions on our financial market.
 In particular, we do not assume any probabilistic structure.

\subsection{An assumption on the European claim}
We consider general path dependent options.
Hence, the pay-off is $X=G(\Sb)$ with any function
$G:\Omega\rightarrow \mathbb{R}$.
Our approach to this problem,
requires us to make the
following regularity
and growth assumption.
Let
 $\|\om\|:=\max_{0\leq k\leq n} |\om_k|$
for $\om \in \Om$.
We assume the following.
\begin{asm}
\label{a.main}
$G$ is upper semi-continuous and bounded
by a quadratic function, i.e.,
there exist a constant $L>0$ such that
$$
|G(\om)| \le L [ 1 + \|\om\|^2],\quad
\forall\ \om \in \Om.
$$
\end{asm}

The above assumption
is quite general and allows for most of the
standard claims such as Asian, lookback, volatility
and Barrier options.
The reason for  the quadratic growth choice
is the volatility options.  More generally, one may
consider different growth conditions as well.
However, in this paper, we choose not
to include this extension
 to avoid more technicalities.

\subsection{Semi static hedging with transaction costs}
Let $\kappa>0$ be a given constant. Consider a model in which
every purchase or sale of the risky asset
at any time is subject to a proportional
transaction cost of rate $\kappa$.
We assume that $\kappa<1/4$.

Then, a {\em{portfolio strategy}} is a pair $\pi:=(f,\gamma)$
where $f:\mathbb{R}_{+}\rightarrow\mathbb{R}$ is a measurable function
and
$$
\gamma:\{0,1,...,N-1\}\times\Omega\rightarrow\mathbb R
$$
is a progressively measurable map, i.e.
$\gamma(i,\om)=\gamma(i,\tilde \om)$ if $\om_j=\tilde \om_j$ for all $j\leq i$.
The function $f$ represents the European option with payoff $f(\Sb_N)$
that is bought at time zero for the price of $\cP(f)$ and
 $\gamma(k,\Sb)$
represents the number of stocks that the investor invests
at time $k$ given that the stock prices
up to time $k$ are  $\Sb_0,\Sb_1,...,\Sb_{k}$.
Then, the
portfolio value at the maturity date is given by
\begin{equation}\label{2.1+}
Y^{\pi}_N(\Sb):=
f(\Sb_N)+\sum_{i=0}^{N-1}
\gamma(i,\Sb)(\Sb_{i+1}-\Sb_i)
-\kappa \sum_{i=0}^{N-1}\Sb_i\left|\gamma(i,\Sb)-\gamma(i-1,\Sb)\right|
\end{equation}
where we set $\gamma(-1,\cdot)\equiv 0$.
The initial cost of any portfolio $(f,\gamma)$
is the price of the option
$\cP(f)$.  Properties of this price operator $\cP$
is given in the next subsection.  

\begin{definition}
\label{d.perfect}
{\rm{A portfolio $\pi$ is called}} perfect {\rm{(or}}
perfectly dominating{\rm{)
if it  super-replicates the option, i.e.,
\begin{equation*}
Y^\pi_N(\Sb)\geq  G(\Sb), \ \ \forall \
\Sb\in\Omega.
\end{equation*}
The minimal}} super-replication cost
{\rm{is given by}}
\begin{equation}
\label{2.1++}
V(G)=\inf\left\{ \cP(f) \ |\
 \pi:=(f,\gamma)\ \mbox{is} \ \mbox{a} \
 \mbox{perfect} \ \mbox{portfolio}\right\}.
\end{equation}
\end{definition}

\subsection{European options and their prices}

We postulate
a general pricing operator  $\cP(f)$ for 
the initial price of the option $f(\Sb_N)$.
We assume that it has the 
following properties.
\begin{asm}
\label{a.kappa}
There exists $p>2$ such that for the power function
$x^p$,
$\cP(x^p)<\infty$. Consider the vector space
$$
\mathcal{H}:=\{f:\mathbb{R}_{+}\rightarrow\mathbb{R}, {\mbox{Borel mbl.}}\ |\ \exists \ C>0 \
{\mbox{such that}} \ |f(x)|\leq C(1+x^p),\  \forall x\in \R_+\}.
$$
We assume that
$\cP:\mathcal{H}\rightarrow\mathbb{R}$ is a convex function and for every constant
$a \in \R$
\begin{equation}
 \label{a.cost}
\cP(a) = a .
\end{equation}
We also assume that $\cP$ is
positively homogeneous of degree one, i.e.
\begin{equation}
 \label{a.additive}
\cP(\lambda f)=\lambda \cP(f), \ \ f\in\mathcal{H}, \ \ \lambda>0.
\end{equation}
Furthermore,
for every sequence ${\{f_n\}}_{n=1}^\infty\subset \mathcal{H}$
converging pointwise to $f\in \mathcal{H}$
 \begin{equation}
 \label{a.cont}
\cP(f) \geq\lim\sup_{n \to \infty} \cP(f_n).
\end{equation}
\end{asm}
In (\ref{2.1++}) we assume that the function $f$ belongs
to $\mathcal{H}$. Namely, $\cP(f)\equiv\infty$ for $f\notin\mathcal{H}$.

We conclude this section
with an elementary result.

\begin{lemma}
\label{l.add} The minimal super-replication
cost $V$ is sub--additive and
positively homogeneous of degree one, i.e.,
$$V(\lambda G)=\lambda V(G), \ \ \lambda>0,$$
and
$$
V(G+H) \le V(G) + V(H).
$$
Furthermore, if $G\geq 0$  and $V(G)<0$,
then $V(G)=-\infty$.
\end{lemma}
\begin{proof}
From the convexity and the positive homogeneity of $\cP$,
it follows that $\cP$ is sub additive, i.e.
$$
\cP(f+g) \le \cP(f) + \cP(g).
$$
Thus, the first two properties follow immediately from (\ref{2.1++}).
Finally, let $G\geq 0$ be a non--negative claim and assume that $V(G)<0$.
Then, there exists a perfect portfolio $(f,\gamma)$
with $\cP(f)<0$. Clearly for any $\lambda>1$, $(\lambda f,\lambda \gamma)$
is also a perfect portfolio. Thus
from (\ref{a.additive}) we get
$$
V(G)\leq \lim_{\lambda\rightarrow\infty}\cP(\lambda f)=-\infty,
$$
as required.
\end{proof}

\subsection{The main result}
To state the main result of the paper,
we need to introduce the probabilistic structure
as well.
Recall the space $\Om$ and
the canonical process $\Sb$.
Let $\F=(\cF_k)_{k=1}^{N}$
be the canonical filtration
generated by the process $\Sb$, i.e,
$\cF_k= \sigma(\Sb_1,\ldots,\Sb_k)$.
\begin{definition}
\label{d.app}
{\rm{
A probability measure $\mathbb Q$  on
$(\Om,\F)$ is called a }}
$\kappa$-approximate martingale law
{\rm{if $\Sb_0=1$
$\mathbb P$-a.s.
and if the pair $(\mathbb Q, \tilde \Sb)$
with
$$
\tilde \Sb_k:=
\E_{\mathbb Q}\left[\  \Sb_N\ |\ \cF_k\right],
$$
is a consistent price system in the sense
of \cite{KS,Sc}, i.e.,
for any  $k<N$
\begin{equation}
\label{2.2-}
(1-\kappa)\Sb_k\leq \tilde \Sb_k
\leq (1+\kappa)\Sb_k \ \ \mathbb{Q}\mbox{-a.s.}
\end{equation}
We denote by
$\mathcal{M}_{\kappa,\cP}$ the set of
all $\kappa$--approximate martingale laws
$\mathbb Q$, such that

\begin{equation}
\label{a.mu}
\E_{\Q}\left[\  f(\Sb_N)\ \right] \le \cP(f),
\qquad
\forall f \in \mathcal{H}.
\end{equation}
}}
\end{definition}

The following theorem is the main result of the paper.
We use the standard convention that the
supremum over an empty set is
equal to
minus infinity.
\begin{theorem}
\label{thm2.1}
Suppose $G$ satisfies
the Assumption \ref{a.main}
and $\cP$ satisfies
the Assumption \ref{a.kappa}. Then,
$$
V(G)=\sup_{\mathbb Q\in
\mathcal{M}_{\kappa,\cP} }\mathbb{E}_{\mathbb Q}
\left[G(\Sb)\right].
$$
In particular, when
the set of measures
$\mathcal{M}_{\kappa,\cP}$ is empty,
$V(G)=-\infty$ for every $G$ satisfying the
Assumption \ref{a.main}
\end{theorem}
\begin{proof}
In view of (\ref{2.1+}) and the convention
$\gamma(-1,\cdot)\equiv 0$,
for  any portfolio $\pi=(f,\gamma)$,
\begin{eqnarray*}
Y^{\pi}_N(\Sb)&=&f(\Sb_N)
+\sum_{i=0}^{N-1}\sum_{j=0}^i(\gamma(j,\Sb)-
\gamma(j-1,\Sb))(\Sb_{i+1}-\Sb_i)
-\kappa \sum_{i=0}^{N-1}\Sb_i
\left|\gamma(i,\Sb)-\gamma(i-1,\Sb)\right|\\
&=&f(\Sb_N)+
\sum_{j=0}^{N-1}
(\gamma(j,\Sb)-\gamma(j-1,\Sb))
\sum_{i=j}^{N-1}(\Sb_{i+1}-\Sb_i)
-\kappa \sum_{i=0}^{N-1}\Sb_i
\left|\gamma(i,\Sb)-\gamma(i-1,\Sb)\right|\\
&=&f(\Sb_N)
+\sum_{j=0}^{N-1}(\gamma(j,\Sb)-\gamma(j-1,\Sb))(\Sb_N-\Sb_j)
-\kappa \sum_{i=0}^{N-1}\Sb_i\left|\gamma(i,\Sb)-\gamma(i-1,\Sb)\right|.
\end{eqnarray*}
Suppose that
$\mathcal{M}_{\kappa,\cP}$ is non-empty.
Let $\mathbb{Q}\in\mathcal{M}_{\kappa,\cP}$ and $\pi=(f,\gamma)$
be a perfect portfolio.  Then,  (\ref{2.2-}) and \reff{a.mu}
yield that
\begin{eqnarray*}
\E_\Q \left[G(\Sb)\right] &\leq&
\E_\Q\left[ Y^{\pi}_N(\Sb)\right]\\
&\le&\cP(f)+
\sum_{i=0}^{N-1}\mathbb{E}_\Q
\left[\left(\gamma(i,\Sb)-\gamma(i-1,\Sb)\right)\
\left(\tilde \Sb_i-\Sb_i\right)\right]\\
&&\hspace{27pt}
-\kappa\sum_{i=0}^{N-1}\mathbb{E}_\Q
\left[\Sb_i \ |\gamma(i,\Sb)-\gamma(i-1,\Sb)|\right]\\
&\leq &\cP(f).
\end{eqnarray*}
So we have proved that
\begin{equation}
\label{e.easy}
\sup_{\mathbb Q\in \mathcal{M}_{\kappa,\cP} }
\E_\Q \left[G(\Sb)\right]  \leq V(G).
\end{equation}

Hence, to complete the proof
of the theorem it suffices to show that
\begin{equation}
\label{2.2}
V(G)\leq\sup_{\mathbb Q\in \mathcal{M}_{\kappa,\cP} }
\mathbb{E}_\Q\left[G(\Sb)\right].
\end{equation}
The proof of the second inequality is
given in the next section.
\end{proof}

\begin{remark}\label{rem2.1}
{\rm{
Consider the following
more general problem.
Assume that for $0<k \le N$ and a set of times $0<i_1<i_2<...<i_k=N$,
one can initially buy vanilla options with
a payoff $f_{i_j}(\Sb_{i_j})$ with maturity date $i_j$
for the price $\cP_{i_j}(f_{i_j})$,
where $\cP_1,...,\cP_k$ satisfy similar assumptions
to Assumption \ref{a.kappa}.
Then, by using the same approach in a recursive manner
we may extend
Theorem \ref{thm2.1} to prove that
the super--replication cost in this context is equal to
$$
\sup_{\mathbb Q\in \mathcal{M}_{\kappa,\cP_1,...,\cP_k} }
\mathbb{E}_\Q\left[G(\Sb)\right]
$$
where $\mathcal{M}_{\kappa,\cP_1,...,\cP_k}$ is
the set of all
$\kappa$-approximate probability laws $\mathbb Q$
and such that for any time $j=1,...,k$ and $f\in\mathcal{H}$ we have
$$\mathbb{E}_{\mathbb Q}f(\Sb_j)\leq\cP_j(f).$$
Furthermore, if the set $\mathcal{M}_{\kappa,\cP_1,...,\cP_k}=\emptyset$
is empty then $V(G)\equiv -\infty$.
For simplicity, in this paper we deal only with the case $k=1$.}}
\end{remark}

\subsection{Fundamental Theorem of Asset Pricing}
\label{ss.ftap}

Theorem \ref{thm2.1} also implies results that can be
seen as the Fundamental Theorem of Asset Pricing
(FTAP) for this market.
Indeed, when the set $\mathcal{M}_{\kappa,\cP}$ of measures
is empty, by Theorem \ref{thm2.1} we conclude that the
minimal super-replication cost of any $G$ (satisfying
the Assumption \ref{a.main}) is equal to minus infinity.  This
is a clear indication of arbitrage.  However, to make a
precise statement, we need to define
the notion of arbitrage.
Since we do not assume a probabilistic structure,
there are at least two possible approaches.
Indeed, in frictionless markets
FTAP is proved under different assumptions
and definitions in \cite{ABS} and in \cite{BN}.
Our result essentially implies FTAP under both definitions
under the Assumption \ref{a.kappa}.

\begin{definition}
\label{d.ftap} {\rm{We say that the model admits}}
\begin{itemize}
\item
no model-independent arbitrage {\rm{ (NA$_{mi}$),
if for every $G\ge 0$ satisfying the Assumption
\ref{a.main}, we have $V(G) \ge 0$.}}
\item
no local arbitrage
{\rm{ (NA$_{local}$), if for every continuous, bounded $G\ge 0, G \not \equiv 0$,
 we have $V(G)> 0$.}}
\end{itemize}
\end{definition}

In the above definition, NA$_{mi}$ is similar to the notion
used in \cite{ABS}.
Also a closely related definition is given by Cox \& Ob{\l}{\'o}j \cite{CO}.
On the other hand NA$_{local}$ is analogous to the one used in \cite{BN}.
One may also consider other versions of NA$_{local}$
by requiring different notions of regularity of $G$.
In the probabilistic setting, this is related to
the choice of the polar sets.  There
one requires that the set $\{G>0\}$
to be non-polar (c.f. \cite{BN}).
Clearly, other choices would result a similar
but a different equivalent condition as proved below.
We do not elaborate on
different choices.

\begin{corollary}
\label{c.ftap}
Suppose $\cP$ satisfies
the Assumption \ref{a.kappa}.
\begin{enumerate}
\item
There is no model-independent arbitrage
if and only if $\mathcal{M}_{\kappa,\cP}$
is non-empty. In particular,
NA$_{mi}$ holds
if and only if there is one $G$ satisfying
the Assumption \ref{a.main} with $V(G)>-\infty$.
\item
There is no local arbitrage
if and only if for every open
subset $O \subset \Omega$
there is $\Q\in \mathcal{M}_{\kappa,\cP}$
with $\Q(O)>0$.
\end{enumerate}
\end{corollary}
\begin{proof}
First statement follows immediately from Theorem \ref{thm2.1}.
So we only prove the second one.
First, assume that NA$_{local}$ holds.  Let
$O \subset \Omega$ be an arbitrary open set.
Set,
$$
G_O(\omega):= \min\{ 1, distance(\omega,\Omega\setminus O)\}.
$$
Since $G_O$ is bounded and continuous,
by NA$_{local}$, $V(G_O) >0$.   Since $0\le G_O \le 1$
and $G_O=0$ outside of $O$,
by Theorem \ref{thm2.1},
$$
0< V(G_O) =
\sup_{\mathbb Q\in
\mathcal{M}_{\kappa,\cP} }\mathbb{E}_{\mathbb Q}
\left[G_O(\Sb)\right]
 \le
\sup_{\mathbb Q\in
\mathcal{M}_{\kappa,\cP} } \ \Q(O).
$$
Hence, there must exists a measure $\Q \in \mathcal{M}_{\kappa,\cP} $
with $\Q(O)>0$.

To prove the opposite implication,
consider a continuous,  bounded option
$G\ge 0, G \not \equiv 0$.
Set
$$
O_G:= \{ \omega \in \Omega\ :\
G(\omega)>0\}.
$$
By the continuity of $G$ , $O_G$ is
a non-empty, open set.
By hypothesis, there exists
 $\Q_G \in \mathcal{M}_{\kappa,\cP} $
with $\Q_G(O_G)>0$.
We estimate using  Theorem \ref{thm2.1}
to arrive at
$$
V(G) = \sup_{\mathbb Q\in
\mathcal{M}_{\kappa,\cP} }\mathbb{E}_{\mathbb Q}
\left[G(\Sb)\right] \ge
\mathbb{E}_{\mathbb Q_G}
\left[G(\Sb)\right] >0.
$$

\end{proof}

\section{Proof of the main result}
\label{sec:3}
\setcounter{equation}{0}
In this section, we prove
(\ref{2.2}).

\subsection{Reduction to bounded uniformly continuous claims}
We first use the elegant path-wise approach of \cite{ABPST}
to martingale inequalities to show that
the super-replication cost of certain options
are asymptotically small.
Indeed, for $M>0$ consider the option,
$$
\alpha_M(\Sb) := \| \Sb\|^2 \ \chi_{\{
\|\Sb\| \geq M\}}.
$$
Let $\Sb^*$ be the running maximum, i.e.,
$$
\Sb^*_k:= \max_{0\leq i\leq k} \Sb_i.
$$
Since $\Sb_k >0$ for each $k$, $\| \Sb\| = \Sb^*_N$.
\begin{lemma}
\label{l.vienna}
$$
\lim_{M \to \infty}
\sup_{\mathbb Q\in \mathcal{M}_{\kappa,\cP} }\mathbb{E}_{\mathbb Q}
\left[\alpha_M(\Sb)\right]\leq\lim_{M \to \infty}
V(\alpha_M) \leq 0.
$$
\end{lemma}
\begin{proof} Let $p>2$ be the exponent in Assumption
\ref{a.kappa}.
Since $\kappa<1/4$ ,
there exists $r\in (2,p)$ such that
$\lambda:=\kappa r c_r<1$,
with
$$
c_r := \frac{r}{r-1}.
$$
We now use Proposition 2.1 in \cite{ABPST}
with the portfolio
$\hat\pi=(\hat f,\hat\gamma)$ given by
$$
\hat{f}(\Sb_N):=
\left(c_r\ \Sb_N\right)^r -c_r,\qquad
\hat\gamma(i,\Sb)=
- r c_r\  \left( \Sb^*_k\right)^{r-1}, \ \ k<N.
$$
We use (\ref{2.1+}) and Proposition 2.1 in \cite{ABPST}
to arrive at
\begin{eqnarray}
\label{3.3}
Y^{\hat\pi}_N(\Sb)&\geq &\| \Sb \|^{r}-\kappa
\sum_{i=0}^{N-1}\Sb_i
\left|\hat\gamma(i,\Sb)-\hat\gamma(i-1,\Sb)\right|\\
\nonumber
&\geq&
\| \Sb \|^{r}-\kappa \|\Sb\|
\sum_{i=0}^{N-1}\left(\hat\gamma(i-1,\Sb)-
\hat\gamma(i,\Sb)\right)=
\|\Sb \|^{r}(1-\lambda).
\end{eqnarray}
Hence,
$$
V\left((1-\lambda) \ \|\Sb\|^r\right) \le \cP( \hat{f}).
$$
Clearly,
$\alpha_M(\Sb) \le \|\Sb\|^r/M^{r-2}$.
Hence, by Lemma \ref{l.add},
\begin{eqnarray*}
V( \alpha_M) &\le& V\left(  \frac{\|\Sb\|^r}{M^{r-2}} \right)
=
 \frac{1}{(1-\lambda)M^{r-2}} V\left((1-\lambda) \ \|\Sb\|^r\right)\\
&\le&
 \frac{1}{(1-\lambda)M^{r-2}} \cP(\hat{f}).
\end{eqnarray*}
Since  $\hat{f}\in\mathcal{H}$,
 $\cP(\hat{f})$ is finite.
Therefore,
$$
\lim_{M \to \infty}
V(\alpha_M)\leq 0.
$$
To complete the proof,
we  recall the proof of \reff{e.easy}
to restate that for every $M$,
$$
\sup_{\mathbb Q\in \mathcal{M}_{\kappa,\cP} }\mathbb{E}_{\mathbb Q}
\left[\alpha_M(\Sb)\right]
\le
V(\alpha_M).
$$
\end{proof}

This result allows us to
consider bounded claims.
We also use a compactness argument,
to obtain the following equivalence.

\begin{theorem}\label{lem3.1}
It suffices to prove \reff{2.2}
for non-negative, bounded, uniformly
continuous claims.
\end{theorem}

Since the proof of this result
is almost orthogonal to the rest of the paper,
we relegate it to the  Appendix.

In view of the above Theorem, in the sequel
we assume that the claim $G$ is non-negative,
bounded and is uniformly continuous.  So we
assume that there exists a constant $K>0$ and
a modulus of continuity, i,e., a continuous
function  $m: \R_+ \to \R_+$
with $m(0)=0$, satisfying,
\begin{equation}
\label{e.buc}
0 \le G(\om) \le K, \qquad
\left| G(\om) - G(\tilde \om) \right| \le m\left(\|\om - \tilde \om\|
\right), \qquad
\forall\ \om, \tilde \om \in \Om.
\end{equation}

If $V(G)=-\infty$ then (\ref{2.2}) is clear. Thus in view of Lemma \ref{l.add}
it follows that without loss of generality, we can assume that
\begin{equation}
\label{e.new}
V(G)\geq 0.
\end{equation}

\subsection{Discretization of the space}
\label{ss.discrete}

Next, we introduce a modification of the original super--replication problem.
Fix $n\in\mathbb{N}$ and set $h=1/n$ and
$U_n=\left\{kh, \ k=0,1,....\right\}$. Denote
$$\mathcal{H}_n:=\{f:U_n\rightarrow\mathbb{R}|\exists C>0 \
\mbox{such} \ \mbox{that} \ |f(x)|\leq C(1+x^p) \ \forall x\}.
$$

For any $g:U_n\rightarrow \mathbb{R}$,
define a function $\cL^{(n)}(g) :\mathbb{R}_{+}\rightarrow\mathbb{R}$ by,
\begin{equation}
\label{3.5}
\cL^{(n)}(g)(x):= (1-\alpha) g\left(\lfloor nx \rfloor h\right)+\alpha
g\left((\lfloor nx \rfloor+1)h\right), \ \ \alpha= nx - \lfloor nx \rfloor,
\qquad x \in \R_+,
\end{equation}
where for a real number $y$, $\lfloor y \rfloor$ is the largest integer
less than or equal to $y$.
Observe that
 $$
 \cL^{(n)} : \mathcal{H}_n \to  \mathcal{H},
 $$
 is a bounded, linear map.

Set $\Omega_n= \left(U_n\right)^N$.
Clearly $\Omega_n\subset \Omega$ and we consider
a financial market where the set of possible
stock price processes is the
set $\Omega_n$.
Then, this  restriction lowers the minimal
super-replication cost.  However,
we restrict the admissible portfolios as well.
Indeed, for a constant $M>0$,
we define the set of {\em{admissible}}
portfolio strategies  below.

\begin{definition}
\label{d.admissible}
{\rm{For any $M>1$,
we say that $\pi:=(g,\gamma)$
is an ($M$-)}}admissible portfolio,
{\rm{ if $g\in\mathcal{H}_n$
and $\gamma:\{0,1,...,N-1\}\times \Omega_n\rightarrow \mathbb{R}$
is a progressively measurable map
sastisfying
$$
|\gamma(i,\Sb)-\gamma(i-1,\Sb)|\leq M,\qquad
\forall \ i\geq 0,\ \Sb \in \Omega_n.
$$
We denote by $\mathcal{A}^n_M$
the set of all admissible portfolios.
A portfolio $\pi\in \mathcal{A}^n_M$ is called}}
 perfect {\rm{ (or perfectly super-replicating) if
\begin{equation*}
Y^\pi_N(\Sb)\geq  G(\Sb), \ \ \forall \
\Sb\in\Omega_n,
\end{equation*}
where $Y^\pi_N$ is given by (\ref{2.1+})}}.
\end{definition}

The minimal super--replication cost is given by
\begin{equation}
\label{3.7}
V^{n,M}(G):=\inf\left\{\cP^{(n)}(g) \ |\ \pi:=(g,\gamma)\in\mathcal{A}^n_M
\  {\mbox{is a perfect portfolio }}\right\},
\end{equation}
where we choose the price function as
\begin{equation}
\label{e.pn}
\cP^{(n)}(g) := \cP\left(\cL^{(n)}(g)\right).
\end{equation}
The following provides the crucial  connection between the original
and the discretized problems. Recall that $h=1/n$.
\begin{proposition}
\label{lem3.2} Assume $G$ satisfies
\reff{e.buc} with a
modulus function $m$.
Then,
for any $M>0$ and $n\in \mathbb{N}$,
$$
V(G)\leq V^{n,M}(G)+
(N + 2 \kappa )MNh+m(h).
$$
\end{proposition}
\begin{proof}
Assume that we have a perfect hedge
$\pi=(g,\gamma) \in \mathcal{A}^n_M$
in the sense of Definition \ref{d.admissible}.
We continue by lifting
this portfolio to a portfolio
$\tilde\pi=( f,\tilde\gamma)$ that is defined
on $\Om$.

Let $ f= \cL^{(n)}(g)$ be as in \reff{3.5}
and define $\tilde \gamma$ by
$$
\tilde\gamma(k,\om):=
\gamma(k,\om_0,\lfloor n \om_1\rfloor h,\ldots,
\lfloor n \om_N\rfloor h), \ \
\forall \ k< N, \   \om=(\om_0,\ldots,\om_N)\in \Omega,
$$
where as before $h=1/n$ and
$\lfloor y \rfloor$ is the integer part of $y$.  Clearly
$\tilde\gamma:\{0,1,...,N-1\}\times \Omega\rightarrow \mathbb R$
is progressively
measurable, and $|\tilde\gamma(i,\Sb)-\tilde\gamma(i-1,\Sb)|\leq M$ for
any $i\geq 0$ and $\Sb \in \Omega$.

For  $\Sb \in \Omega$ define $\Sb^{(1)},\Sb^{(2)}$ by
$$
\Sb^{(1)}_k:=\lfloor n \Sb_k \rfloor h, \ \
\Sb^{(2)}_k:= \Sb^{(1)}_k+ h \delta_k^N,
\ \ {\mbox{for all}}\  k\leq N,
$$
where $ \delta_k^N$ is equal to one when
$k=N$ and zero otherwise.
Then,
there exists $\lambda\in [0,1]$ such that $\Sb_N=
\lambda \Sb^{(1)}_N+(1-\lambda)\Sb^{(2)}_N$.
Also both $\| \Sb^{(1)}-\Sb \|$ and $\| \Sb^{(2)}-\Sb \|$
are less than $h=1/n$.  Moreover,
$\gamma(k,\Sb^{(1)})=\gamma(k,\Sb^{(2)})=\tilde\gamma(k,\Sb)$
for every $k <N$.
We use these
together with (\ref{2.1+}), (\ref{3.5}) and the fact that $\gamma\in [-M N,M N]$.
The result is
\begin{eqnarray}
\label{3.10}
Y^{\tilde\pi}_N(\Sb)&\geq& \lambda Y^\pi_N(\Sb^{(1)})
+(1-\lambda)Y^\pi_N(\Sb^{(2)})- (N+2\kappa)MNh\\
\nonumber
&\geq&\lambda G(\Sb^{(1)})+(1-\lambda)G(\Sb^{(2)})-(N+2\kappa)MNh \\
&\geq& G(\Sb)-m\left(h\right)-(N+2\kappa)MNh,
\nonumber
\end{eqnarray}
where the last inequality follows from  \reff{e.buc}.
Thus $(f+m\left(h\right)+(N+2\kappa)MNh,\tilde\gamma)$
is a perfect hedge
in the sense of Definition \ref{d.perfect}. This together with the equality
$\cP_n(g)=\cP(f)$
completes the proof.
\end{proof}

\subsection{Analysis of $V^{n,\sqrt {n}}(G)$}
\label{ss.analysis}
Fix $n>0$.
From Proposition \ref{lem3.2} and (\ref{e.new}), it follows that for sufficiently
large $n$
\begin{equation}
\label{e.neww}
V^{n,\sqrt {n}}(G)\ge -1.
\end{equation}

Fix $n\in\mathbb{N}$ sufficiently large such that (\ref{e.neww})
holds true.
We introduce three auxiliary sets.
Let $\cW_n$ be the set of all functions $g\in \mathcal{H}_n$
which satisfy the growth condition
$$
\|g\|_* : = \sup_{\{x \in U_n\}} \ \left\{
\frac{|g(x)|}{(1+x)^p}\ \right\} \le n.
$$
Let
${\cQ}_n$ be the set of  all probability measures
 $\mathbb Q$ on $\Omega_n$ which satisfy
$$
\mathbb{E}_\Q \left[\|\Sb\|^p\right]<\infty.
$$
Finally, let $\hat \cQ_n$ be
set of all probability measures $\mathbb Q \in \cQ_n$ which satisfy
\begin{equation}
\label{3.10+}
\E_\Q\left[\ g(\Sb_N)\ \right] \le \cP^{(n)}(g)
+ \frac{K+1}{n}\ \| g\|_*,
\qquad
\forall g \in \cW_n,
\end{equation}
where
$$
K:= \sup_{\Sb \in \Om} \ G(\Sb).
$$
We shall show in the below proof that
in view of \reff{e.neww}, the set of measures
$\hat \cQ_n$ is non-empty
for all sufficiently large $n$.

The following provides an upper bound for the
super-replication cost $V^{n,\sqrt {n}}$ defined
by (\ref{3.7}).
\begin{lemma}
\label{lem3.3}
Suppose that $G$ satisfies
\reff{e.buc}--(\ref{e.new}).  Then,
for all sufficiently large $n$,
$$
V^{n,\sqrt {n}}(G)\leq \sup_{\mathbb  Q\in {\hat \cQ}_n}
\mathbb{E}_\Q\left(G(\Sb)-\sqrt{n}\sum_{k=0}^{N-1}
\left(\left|\mathbb{E}_\Q(\Sb_N|\cF_k)-\Sb_k\right|-\kappa \Sb_k\right)^{+}\right).
$$
\end{lemma}
\begin{proof}
Define
$H:\cW_n\times \cQ_n\rightarrow \mathbb{R}$ by
$$
H(g,\Q):=
\E_\Q\left(G(\Sb)-g(\Sb_N)-\sqrt{n}\sum_{k=0}^{N-1}
\left(|\E_\Q[\Sb_N|\cF_k]-\Sb_k |-\kappa \Sb_k\right)^{+}\right)+
\cP^{(n)}(g).
$$
Since $\cP^{(n)}$ is finite on $\cW_n$,
in view of  the
definitions of  $\cW_n$ and $\cQ_n$,
$H$ is well defined.
We now use Theorem \ref{thm4.1}
that will be proved in the next section
with $F(\Sb):=G(\Sb)-g(\Sb_N)$
with an arbitrary $g \in \cW_n$.
This yields that
$$
V^{n,\sqrt {n}}(G)\leq \sup_{\mathbb Q\in \cQ_n}
H(g,\mathbb Q), \qquad
\forall \ g \in \cW_n.
$$
Hence,
\begin{equation}\label{3.11}
V^{n,\sqrt {n}}(G)\leq\inf_{g\in W_n}\sup_{\mathbb Q\in
\cQ_n}H(g,\mathbb Q).
\end{equation}
Since the functions in $\cW_n$ are restricted
to satisfy the growth condition, the above
is possibly an inequality and not an equality.

Next, we continue by interchanging the order
of the above infimum and supremum.
For that purpose, consider
the vector space $\mathbb{R}^{U_n}$ of all functions
$g:U_n\rightarrow\mathbb{R}$ induced with the topology of
point-wise
convergence. This space is locally convex
and since $U_n$ is countable,
$\cW_n\subset \mathbb{R}^{U_n}$ is compact.
Also, the set $\cQ_n$ can be naturally considered as a convex subspace of
the vector space $\mathbb{R}^{\Omega_n}$.
In order to apply a min-max theorem,
we need to show continuity and concavity.
In view of Assumption \ref{a.kappa},
in the first variable $H$ is convex and is therefore
continuous due to the dominated convergence theorem.
We also claim
that $H$ is concave in the second
variable. For this, it is sufficient to show that for
any $k<N$ the functional
$\mathbb{E}_\Q\left(|\mathbb{E}_\Q(\Sb_N |\cF_k)
-\Sb_k|-\kappa \Sb_k\right)^{+}$
is convex in $\mathbb{Q}$.
Indeed, this convexity follows
from the following representation
\begin{eqnarray*}
&&
\mathbb{E}_\Q\left(|\mathbb{E}_\Q(\Sb_N|\cF_k)
-\Sb_k|-\kappa \Sb_k\right)^{+}
\\
&=&\sum_{(z_1,...,z_k)\in U^k_n}
\left(\left |\sum_{z_N\in U_n} z_N
\Q(A(z_1,\ldots,z_k,z_N)) - z_k \Q(B(z_1,\ldots,z_k)) \right| - \kappa
z_k \Q(B(z_1,\ldots,z_k))\right)^+
\end{eqnarray*}
where
$$
A(z_1,\ldots,z_k,z_N)= \left\{
\Sb_1=z_1,\ldots,\Sb_k=z_k,\Sb_N=z_N\right\},
\quad
B(z_1,\ldots,z_k)= \left\{
\Sb_1=z_1,\ldots,\Sb_k=z_k\right\}.
$$
We now apply Theorem 2 in \cite{BHLP}
to the function $H$.  The result is
$$
\inf_{g\in \cW_n}\sup_{\mathbb Q\in {\cQ}_n}
H(g,\mathbb Q)=
\sup_{\mathbb Q\in {\cQ}_n}\inf_{g\in \cW_n}H(g,\mathbb Q).
$$
We combine the above inequality with the previous one
to obtain,
\begin{equation}
\label{3.12}
V^{n,\sqrt {n}}(G) \le
\sup_{\mathbb Q\in {\cQ}_n}\inf_{g\in \cW_n}H(g,\mathbb Q).
\end{equation}

Now suppose that $\Q \in  {{\cQ}_n}$
but not in $ {\hat \cQ}_n$. Then,
there is $g^* \in \cW_n$ so that
$$
\E_\Q\left[\ g^*(\Sb_N)\ \right] > \cP^{(n)}(g^*) +
\frac{K+1}{n}\ \|g^*\|_*.
$$
By the positive homogeneity
of $\cP$,
we may assume that $\|g^*\|_*=n$.  Then,
$$
\E_\Q\left[\ g^*(\Sb_N)\ \right] > \cP^{(n)}(g^*) + K+1,
$$
and recall that $K= \sup_{\Omega} G$.
The definition of $H$ yields that
$$
 H(g^*,\Q) \le \E_\Q\left[G(\Sb)\right]
-\E_\Q\left[\ g^*(\Sb_N)\ \right] + \cP^{(n)}(g^*)
<  \E_\Q\left[G(\Sb)\right] -K-1 \le -1.
$$
In view of \reff{e.neww}, we conclude
that there must exists measures in
${\hat \cQ}_n$.  Additionally,
we may restrict the maximization in \reff{3.12}
over the probability measures $\Q \in {\hat \cQ}_n$.
We use this restricted version of \reff{3.12} to arrive at
$$
V^{n,\sqrt {n}}(G) \le
\sup_{\Q \in {\hat \cQ}_n}\inf_{g\in \cW_n}H(g,\mathbb Q)
\le \sup_{\Q \in {\hat \cQ}_n} H(0,\mathbb Q).
$$
Since $\cP(0)=0$, the above is exactly
the statement of the lemma.
\end{proof}

\subsection{Last step of the proof}
We  combine Proposition
\ref{lem3.2} and Lemma \ref{lem3.3}
to conclude that
\begin{equation}
\label{3.12+}
V(G)\leq \lim\inf_{n\rightarrow\infty} \beta_n,
\end{equation}
where
$$
\beta_n:=
\sup_{\mathbb Q\in {\hat \cQ}_n}
\mathbb{E}_\Q\left(G(\Sb)-
\sqrt n\sum_{k=0}^{N-1}
\left(|\mathbb{E}_\Q
(\Sb_N|\cF_k)-\Sb_k|-\kappa \Sb_k\right)^{+}\right).
$$
Thus in order to complete the proof of inequality (\ref{2.2})
 it is sufficient to establish the following.
\begin{lemma}
\label{lem3.4}
Suppose that $G$ satisfies \reff{e.buc}--(\ref{e.new}).  Then,
\begin{equation}
\label{3.13}
\liminf_{n\rightarrow\infty}
\ \beta_n \le \sup_{\mathbb Q\in \mathcal{M}_{\kappa,\cP} }
\mathbb{E}_\Q [G(\Sb)].
\end{equation}
\end{lemma}
\begin{proof}
From (\ref{e.new}) and (\ref{3.12+}) it follows that for sufficiently
large $n$, $\beta_n \ge -1$.
Therefore, for
all sufficiently large $n\in\mathbb{N}$,
there exists  $\Q_n\in \hat{\cQ}_n$
so that
\begin{equation}
\label{3.14}
\mathbb{E}^{(n)}\sum_{k=0}^{N-1}
\left(\left(|\mathbb{E}^{(n)}(\Sb_N|
\cF_k)-\Sb_k|-\kappa \Sb_k\right)^{+}\right)\leq
 \frac{K+1}{\sqrt n}
\end{equation}
and
\begin{equation}
\label{3.15}
\beta_n \le \frac{1}{n}+ \mathbb{E}^{(n)}\left(G(\Sb)-\sqrt n\sum_{k=0}^{N-1}
\left(|\mathbb{E}^{(n)}(\Sb_N|\cF_k)-\Sb_k|-\kappa \Sb_k\right)^{+}\right),
\end{equation}
where $\mathbb E^{(n)}$ denotes the expectation
with respect to $\Q_n$.
From (\ref{e.pn}) and (\ref{3.10+}) we get
\begin{equation*}
\mathbb{E}^{(n)}\left[ \Sb^p_N \right] \le
\cP^{(n)}\left(x^p\right) +\frac{K+1}{n} \| x^p\|_*\leq
\cP\left(2+2x^p\right) +K+1<\infty.
\end{equation*}
Hence,
\begin{equation}
\label{3.15+}
\sup_{n\in\mathbb N} \mathbb{E}^{(n)}
\left[ \Sb^p_N \right]<\infty.
\end{equation}

We claim
that the  probability measures
$\Q_n$, $n\in\mathbb N$ are tight.
Indeed, in view of the uniform second moment estimate \reff{3.15+},
tightness would follow from the
uniform integrability which states that
for any $A>0$,
$$
\lim_{A\rightarrow\infty}\
\sup_{n\in\mathbb N}\
\mathbb E^{(n)}(\Sb_k\ \chi_{\{\Sb_k>A\}})=0,
\qquad
\forall \  k=1,...,N-1.
$$
Since $\Sb_k$ is $\Q_{n}$ integrable,
the above would follow from
\begin{equation}
\label{3.15++}
\lim_{M\to \infty}\ \lim_{A\rightarrow\infty}\
\sup_{n \ge M}\
\mathbb E^{(n)}(\Sb_k\ \chi_{\{\Sb_k>A\}})=0,
\qquad
\forall \  k=1,...,N-1.
\end{equation}

We continue by proving \reff{3.15++}.
Fix positive integers  $k<N$ and $n$.
Set $X:=(1-\kappa) \Sb_k$,
$Y:=\mathbb E^{(n)}(\Sb_N|\cF_k)$.
In view of (\ref{3.14}),
 $\mathbb E^{(n)}((X-Y)^{+})\leq (K+1)/\sqrt n$.
Therefore, by  Cauchy-Schwarz
and the Markov inequality, we obtain
that for any $A>0$,
\begin{eqnarray*}
\mathbb E^{(n)}[X\ \chi_{\{X>A\}}]
&\leq &\mathbb E^{(n)}\left[((X-Y)^++Y) \chi_{\{X>A\}}\right]
\le \mathbb E^{(n)}\left[(X-Y)^+\right]
+\E^{(n)}\left[Y\ \chi_{\{X>A\}}\right]\\
&\leq &
\frac{K+1}{\sqrt n}+
\mathbb E^{(n)}\left[Y\ \chi_{\{X>A\}}\right]
\leq
 \frac{K+1}{\sqrt n}+\sqrt{\mathbb E^{(n)} [Y^2]}\
\sqrt{\Q_n(X>A)}\\
&\leq &
\frac{K+1}{\sqrt n}+\sqrt{\mathbb E^{(n)}[ \Sb_N^2]}\
\sqrt{\mathbb E^{(n)} [X]} \frac{1}{\sqrt{A}}\\
&\leq&
\frac{K+1}{\sqrt n}+  \frac{1}{\sqrt{A}}
\sqrt{\mathbb E^{(n)}[ \Sb_N^2]}\
\sqrt{\frac{K+1}{\sqrt n} + \mathbb E^{(n)} [Y]}\\
&\leq&
\frac{K+1}{\sqrt n}+  \frac{1}{\sqrt{A}}
\sqrt{\mathbb E^{(n)} [\Sb_N^2]}\
\sqrt{\frac{K+1}{\sqrt n} + \mathbb E^{(n)} [\Sb_N]}.
\end{eqnarray*}
This together with \reff{3.15+} yields  \reff{3.15++}
and hence, the uniform integrability of the
sequence $\Q_{n}$.

In view of the Prohorov's Theorem (see \cite{B}),
there exists a
subsequence  $\Q_{n_l}$,
$l\in\mathbb N$ which converge weakly
 to a probability measure $\tilde{\Q}$.
Then, (\ref{3.15}) implies that
\begin{eqnarray*}
\tilde{\mathbb E} G(\Sb)&=&\lim_{l\rightarrow\infty}\mathbb{E}^{(n_l)} G(\Sb)\\
&\geq&
\lim\inf_{n\rightarrow\infty}\sup_{\mathbb Q\in {\hat \cQ}_n}
\mathbb{E}_\Q\left(G(\Sb)-\sqrt n\sum_{k=0}^{N-1}
\left(|\mathbb{E}_\Q(\Sb_N|\cF_k)-\Sb_k|-\kappa \Sb_k\right)^{+}\right),
\end{eqnarray*}
where $\tilde {\mathbb E}$ denotes the expectation with
respect to $\tilde {\mathbb{Q}}$.  Then, Proposition \ref{lem3.2} and
Lemma \ref{lem3.3} imply \reff{2.2} provided
that $\tilde {\mathbb{Q}} \in \cM_{\kappa,\cP}$.

Thus, in order to complete the proof of this lemma,
 it suffices to show that for the limiting probability
measure $\tilde{\Q}$ is in $\cM_{\kappa,\cP}$.

Fix $k$ and
let $h:\mathbb{R}^k\rightarrow\mathbb{R}_{+}$
be a continuous bounded function. Denote by $\|\cdot\|$
the sup norm on $\mathbb R^k$. By \reff{3.14}, it follows that
\begin{eqnarray*}
\mathbb E^{(n)}(\Sb_N h(\Sb_1,\ldots,\Sb_k))&=&
\mathbb E^{(n)}\left(\mathbb E_n(\Sb_N|\cF_k)h(\Sb_1,\ldots,\Sb_k)\right)\\
&\leq&
\mathbb E^{(n)}\left((1+\kappa)\Sb_k h(\Sb_1,\ldots,\Sb_k)\right)+\frac{(K+1)\|h|\|}{\sqrt n}.
\end{eqnarray*}
Similarly, we conclude that
\begin{eqnarray*}
\mathbb E^{(n)}(\Sb_N h(\Sb_1,\ldots,\Sb_k))
&=&\mathbb E^{(n)}\left(\mathbb E_n(\Sb_N|\cF_k)h(\Sb_1,\ldots,\Sb_k)\right)\\
&\geq&
\mathbb E^{(n)}\left((1-\kappa)\Sb_k h(\Sb_1,\ldots,\Sb_k)\right)-\frac{(K+1)\|h\|}{\sqrt n}.
\end{eqnarray*}
We next
take the limit $n_l\rightarrow\infty$, and
use \reff{3.15+}, \reff{3.15++}.
The result is
\begin{equation*}
(1-\kappa)\tilde{\mathbb E}\left(\Sb_k h(\Sb_1,\ldots,\Sb_k)\right)\leq
\tilde{\mathbb E}\left(\Sb_N h(\Sb_1,...,\Sb_k)\right)\leq
(1+\kappa)\tilde{\mathbb E}\left(\Sb_k h(\Sb_1,\ldots,\Sb_k)\right).
\end{equation*}
The above holds for any
non-negative, continuous and bounded function $h$.
Then, by a standard density argument we arrive at
$$
(1-\kappa)\Sb_k\leq
\tilde{\mathbb E} (\Sb_N|\cF_k)\leq (1+\kappa) \Sb_k, \ \ k=0,...,N-1.
$$
Hence, $\tilde \Q$ is an $\kappa$-approximate martingale law.

We continue by showing that $\tilde \Q$ satisfies \reff{a.mu}.
From (\ref{3.15+}) it follows
that
\begin{equation}
\label{3.17}
\tilde{\mathbb{E}}
\left[ \Sb^p_N \right]<\infty.
\end{equation}
Let $g \in \mathcal{H}$ and let
$C>0$ be such that $g(x)\leq C(1+x)^p$ for all $x\geq 0$.
 There exists a sequence of continuous functions
$\{g_k\}_{k=1}^\infty\subset\mathcal{H}$ which convergence pointwise
to $g$ and satisfy $|g_k(x)|\leq C(1+x^p)$, for all $x\geq 0$ and $n\in\mathbb{N}$.
Moreover, by \reff{3.17} and the dominated convergence theorem,
$$
\tilde{\mathbb{E}}[g(\Sb_N)]=
\lim_{k\rightarrow\infty}\tilde{\mathbb{E}}[g_k(\Sb_N)].
$$
Thus, to prove $\tilde{\mathbb{E}}[g(\Sb_N)]\leq\cP(g)$,
it is sufficient to show
that $\tilde{\mathbb{E}}[g_n(\Sb_N)]\leq\cP(g_n)$ for any $n$.
Therefore,  without loss of generality
we may assume that $g$ is a continuous function. Set
$f_n:=g_{|U_n}$ and $h_n=\cL_n(f_n)$, $n\in\mathbb{N}$. Observe that
for sufficiently large $n$, $f_n\in\cW_n$.
Since $g$ is continuous,
$g(x)=\lim_{n\rightarrow\infty}f_n(x_n)$
for any $x\geq 0$ and a sequence $x_n\geq 0$, $n\in\mathbb{N}$
which converge to $x$.
Furthermore the sequence $h_n$, $n\in\mathbb{N}$ convergence pointwise to $g$.
We use the Skorohod representation theorem,
(\ref{a.cont}) and (\ref{e.pn}), to conclude that
$$
\tilde{\mathbb E}[g(\Sb_N)]=\lim_{n\rightarrow\infty}
\mathbb E^{(n)}[f_n(\Sb_N)]
\leq \lim_{n\rightarrow\infty} \cP_n(f_n)
=\lim_{n\rightarrow\infty} \cP(h_n)\leq \cP(g)
$$
as desired.

\end{proof}

\section{Hedging with Constraints and Transaction costs}
\label{sec:4}
\setcounter{equation}{0}

This section is devoted to the proof
of an auxiliary result that is used in Lemma \ref{lem3.3}.

Fix $n\in\mathbb{N}$ and
recall $\Om_n=\{\ kh\ |\ k=0,1,\ldots\ \}$
with $h=1/n$ as defined
in the subsection \ref{ss.discrete}.
In this section, we do not allow to buy vanilla options,
but only to trade the stock with
proportional transaction costs. Furthermore, the number of the
stocks that the investor is allowed
to buy should lie in the interval $[-M,M]$.
Therefore,
 in this section a portfolio strategy is a pair
 $\tilde\pi=(x,\gamma)$ where $x\in\mathbb{R}$ is the
initial capital and $\gamma:\{0,1,...,N-1\}\times \Omega_n\rightarrow\mathbb R$
is a progressively measurable map
which satisfy $|\gamma(i,S)-\gamma(i-1,S)|\leq M$ for all $i,S$.
The portfolio value for any
$\Sb \in\Omega_n$ is given by
$$
\tilde{Y}^{\tilde\pi}_N(\Sb)=
x+\sum_{i=0}^{N-1}\gamma(i,\Sb)(\Sb_{i+1}-\Sb_i)
-\kappa \sum_{i=0}^{N-1}\Sb_i\left|\gamma(i,\Sb)-\gamma(i-1,\Sb)\right|,
$$
where as before we set $\gamma(-1,S)\equiv 0$.

Consider a European option with the payoff $\hat X=F(\Sb)$ where
$F:\Omega_n\rightarrow\mathbb{R}$.
We do not make any assumptions on the function $F$.
The super--replication price is defined by
$$
\tilde{V}(F)=\inf\{x\ |\ \exists \tilde\pi=(x,\gamma) \ \mbox{such} \ \mbox{that} \
\tilde{Y}^{\tilde\pi}_N(\Sb)\geq F(\Sb),
 \ \forall \ \Sb\in\Omega_n\}.
 $$

\begin{theorem}
\label{thm4.1}
For any $F:\Omega_n\rightarrow\mathbb{R}$,
$$
\tilde{V}(F)=\sup_{\mathbb Q\in \tilde{\cQ}_n}
\mathbb{E}_\Q
\left(F(\Sb)-M\sum_{k=0}^{N-1}\left(|\mathbb{E}_\Q
(\Sb_N|\cF_k)-\Sb_k|-\kappa \Sb_k\right)^{+}\right),
$$
where $\tilde{\cQ}_n$ is the set of all
probability measures on $\Omega_n$,
which  are supported on a finite set.
\end{theorem}
\begin{proof}
We start with establishing the inequality
\begin{equation}
\label{4.1}
\tilde{V}(F) \le \sup_{\mathbb Q\in \tilde{\cQ}_n}
\mathbb{E}_\Q
\left(F(\Sb)-M\sum_{k=0}^{N-1}\left(|\mathbb{E}_\Q
(\Sb_N|\cF_k)-\Sb_k|-\kappa \Sb_k\right)^{+}\right).
\end{equation}
In fact in Lemma \ref{lem3.3} we used only
the above inequality. Without loss of generality we assume
that the right hand side of (\ref{4.1}) is finite.

For a positive integer $J\in\mathbb N$,
consider the finite set
$\Omega^J_n:=\{0,h,2 h,\ldots, Jh\}^N$
with as before $h=1/n$.
Define the minimal super--replication cost
$$
\tilde{V}^J(F)=
\inf\{x\ |\ \exists \tilde\pi=(x,\gamma) \ \mbox{such} \ \mbox{that} \
\tilde{Y}^{\tilde\pi}_N(\Sb)\geq F(\Sb),
 \ \forall \Sb\in\Omega^J_n\}.
 $$
The cost $\tilde{V}^J(G)$ is in fact
equal to
the minimal super-replication cost
in the multinomial model
which is supported on the set $\Omega^J_n$.
Thus, we are in a position to apply
Theorem 3.1 in \cite{DS1} with the penalty function
\begin{equation}\label{4.3}
g(\tilde s,\nu) =
\left\{
\begin{array}{ll}
\kappa \tilde s |\nu|,\qquad &{\mbox{if}}\ |\nu| \le {M},\\
+\infty, &{\mbox{else}}.
\end{array}
\right.
\end{equation}
The function $g$ is convex in the second variable.
Moreover, the convex dual of $g$
is given by
$$
\hat G(\tilde s,y)=\sup_{\nu\in\mathbb{R}}\nu y-g(\tilde s,\nu)
=M(|y|-\kappa s)^{+}.
$$
Therefore, Theorem 3.1 in \cite{DS1}
implies that
\begin{eqnarray}
\label{4.1+}
\tilde{V}^J(F)&=&
\sup_{\mathbb Q\in \cQ^J_n}
\mathbb{E}_\Q
\left(F(\Sb)-
M\sum_{k=0}^{N-1}
\left(|\mathbb{E}_\Q
(\Sb_N|\cF_k)-\Sb_k|-\kappa \Sb_k\right)^{+}\right)\\
&\leq&
\sup_{\mathbb Q\in \tilde{\cQ}_n}
\mathbb{E}_\Q
\left(F(\Sb)-
M\sum_{k=0}^{N-1}
\left(|\mathbb{E}_\Q
(\Sb_N|\cF_k)-\Sb_k|-\kappa \Sb_k\right)^{+}\right),\nonumber
\end{eqnarray}
where $\cQ^J_n$ is the set of all probability measures
on $\Omega^J_n$.

Now, for every
$J\in\mathbb{N}$ there exists a super-replicating
portfolio
$\tilde\pi_J=\left(\tilde{V}^J(F)+1/J,\gamma_J\right)$
for the multinomial model supported on $\Omega^J_n$. Namely,
$\gamma_J:\{0,1,...,N-1\}\times \Omega_n\rightarrow\mathbb{R}$
is a progressively measurable map
such that $|\gamma_J(i,\Sb)-\gamma_J(i-1,\Sb)|\leq M$ for any $i,\Sb$
and
$\tilde{Y}^{\tilde{\pi}_J}_N(\Sb)\geq F(\Sb)$,
for every  $\Sb\in\Omega^J_n.$
 By using standard a diagonal procedure,
 we construct
 a subsequence ${\{\gamma_{J_i}\}}_{i=1}^\infty$
 such that for any $j=0,1,...,N-1$ and $\Sb \in\Omega_n$,
 $\lim_{i\rightarrow\infty}\gamma_{J_i}(j,\Sb)$
 exists. We denote this limit by $\gamma(j,\Sb)$.
 Let $x=\lim\inf_{i\rightarrow\infty}\tilde{V}^{J_i}(F)$.
 Then, clearly $\gamma\{0,1,...,N-1\}\times \Omega_n\rightarrow\mathbb{R}$
 is a progressively measurable map
and the portfolio which is given by
$\tilde\pi=( x,\gamma)$ satisfy
 $|\gamma_J(i,\Sb)-\gamma_J(i-1,\Sb)|\leq M$ for any $i$, $\Sb$.
 Moreover,
$\tilde{Y}^{\tilde{\pi}}_N(\Sb)\geq F(\Sb)$,
for every $\Sb\in\Omega_n.$ This together with (\ref{4.1+}) yields that
 $$
 \tilde{V}(F)\leq x\leq\sup_{\mathbb Q\in \tilde{\cQ}_n}
\mathbb{E}_\Q
\left(F(\Sb)-M\sum_{k=0}^{N-1}
\left(|\mathbb{E}_\Q(\Sb_N|\cF_k)-\Sb_k|-\kappa \Sb_k\right)^{+}\right),
$$
and (\ref{4.1}) follows.

Finally,  by using similar arguments to the arguments
on page 9 in \cite{DS1}, we prove the inequality
$$
\tilde{V}(F)\geq\sup_{\mathbb Q\in \tilde{\cQ}_n}
\mathbb{E}_\Q
\left(F(\Sb)-M\sum_{k=0}^{N-1}
\left(|\mathbb{E}_\Q(\Sb_N|\cF_k)-\Sb_k|-\kappa \Sb_k\right)^{+}\right),
$$
and complete the proof.
\end{proof}

\section{Appendix}

In this appendix, we prove Theorem \ref{lem3.1}.
We proceed in several lemmas.  We first use Lemma \ref{l.vienna}
to reduce the problem to bounded claims.
Then, using a compactness
argument as in \cite{BHLP},
we further reduce it to bounded and continuous claims.

\begin{lemma}  Suppose \reff{2.2} holds for all
bounded, upper semi-continous continuous functions.
Then, it also holds
for all $G$ satisfying Assumption \ref{a.main}.
\end{lemma}
\begin{proof}
Suppose that $G$ satisfies Assumption \ref{a.main}.
Let $\varphi$ be any smooth function satisfying
$$
0 \le \varphi \le 1, \quad
\varphi(\Sb) =1, \ \ \forall \ \| \Sb \| \le 1,\quad
\varphi(\Sb) =0, \ \ \forall \ \| \Sb \| \ge 2.
$$
For a constant $M>1$, set
$$
\varphi_M(\Sb):= \varphi(\Sb/M),\qquad
G_M := G \varphi_M.
$$
$G_M$ is bounded and upper semi-continuous.
Then, by the hypothesis, the inequality \reff{2.2}
and the duality formula holds for $G_M$.
In view of Assumption \reff{a.main},
$$
\left|G(\Sb) - G_M(\Sb)\right| \le  L(1+\|\Sb\|^2)
\chi_{\{ \| \Sb \| \ge M\}}.
$$
Let $\alpha_M$ be as in
Lemma \ref{l.vienna}.  Then, for all sufficiently
large $M$,
$$
\left|G(\Sb) - G_M(\Sb)\right| \le 2L \alpha_{M}(\Sb).
$$
Since $G_M$ satisfies \reff{2.2},
$$
V(G_M) \le \sup_{\mathbb Q\in \mathcal{M}_{\kappa,\cP} }\mathbb{E}_{\mathbb Q}
\left[G_M(\Sb)\right]
\le \sup_{\mathbb Q\in \mathcal{M}_{\kappa,\cP} }\mathbb{E}_{\mathbb Q}
\left[G(\Sb)\right] + 2L \sup_{\mathbb Q\in \mathcal{M}_{\kappa,\cP} }\mathbb{E}_{\mathbb Q}
\left[\alpha_M(\Sb)\right]  .
$$
By the subadditivity of
the minimal super-replication cost $V$,
$$
V(G) \le V(G_M) + 2L \ V\left(\alpha_{M}\right).
$$
Combining the above inequalities
and Lemma \ref{l.vienna}, we arrive at
\begin{eqnarray*}
V(G) &\le & \liminf_{M\to \infty}\ \left[
 V(G_M) + 2L\ V\left(\alpha_{M}\right)\right]\\
 &\le&\sup_{\mathbb Q\in \mathcal{M}_{\kappa,\cP} }\mathbb{E}_{\mathbb Q}
\left[G(\Sb)\right] +  2L \  \liminf_{M\to \infty}
\left[  V\left(\alpha_{M}\right) + \sup_{\mathbb Q\in \mathcal{M}_{\kappa,\cP} }\mathbb{E}_{\mathbb Q}
\left[\alpha_M(\Sb)\right] \right] \\
 &\leq&\sup_{\mathbb Q\in \mathcal{M}_{\kappa,\cP} }\mathbb{E}_{\mathbb Q}
\left[G(\Sb)\right] .
\end{eqnarray*}

\end{proof}

The above proof also yields the following equivalence.

\begin{lemma}  Suppose \reff{2.2} holds for all, non-negative,
bounded, uniformly  continuous functions.
Then, it also holds
for all $G$ that are bounded
and continuous.
\end{lemma}
\begin{proof}
Let $G$ be a bounded continuous function.
By adding $G$ an appropriate constant,
we may assume that it is nonnegative as well.
Given an integer $N$,
define $G_N$ as before.  Since $G_N$ is
compactly supported and continuous,
it is also uniformly continuous.
We then proceed exactly as in the previous
lemma to conclude the proof.
\end{proof}

We need the following elementary result.
\begin{lemma}
Let $G$ be bounded and upper semicontinuous.
Then, there exists a sequence of
uniformly
bounded, continuous functions $G_n:\R_+^d \to \R$, so that
$G_n \ge G$ and
\begin{equation}\label{app.1}
\lim\sup_{n\rightarrow\infty}G_n(x_n)\leq G(x),
\end{equation}
for every $x\in\mathbb{R}^d_{+}$ and every
sequence  $\{x_n\}_{n=1}^\infty\subset \mathbb{R}^d_{+}$
with $\lim_{n\rightarrow\infty} x_n=x$.
\end{lemma}
\begin{proof}
For a positive integer $n$,
consider the grid
$O_n=\left\{\left(\frac{k_1}{n},...,\frac{k_d}{n}\right),
k_1,...,k_d\in \mathbb{Z}_{+}\right\}$. Define the function $G_n:O_n\rightarrow\mathbb{R}_{+}$
by
$$
G_n(x)=\sup_{\{u\in\mathbb{R}^d_{+}\ |\ \|u-x\|\leq \frac{2}{n}\}} G(u), \qquad x\in O_n.
$$
Next, we extend $G_n$ to the domain $\mathbb{R}^d_{+}$.

For any $k_1,...,k_d\in\mathbb{Z}_{+}$ and
a permutation $\sigma:\{1,...,d\}\rightarrow\{1,...,d\}$
consider the $d$--simplex
$$
U^\sigma_{k_1,...,k_d}=\left\{(x_1,...,x_d): \frac{k_i}{n}\leq x_i\leq \frac{k_i+1}{n}, \ i=1,...,d \right\}
\bigcap \left\{(x_1,...,x_d): x_{\sigma(i)}\leq x_{\sigma(j)}, \forall i<j\right\}.$$
Fix a simplex $U^\sigma_{k_1,...,k_d}$.
Any $u\in U^\sigma_{k_1,...,k_d}$ can be represented uniquely as
a convex combination of the simplex vertices $u_1,...,u_{d+1}$ (which belong to $O_n$).
Thus define a continuous function
$G^{n,\sigma}_{k_1,...,k_d}:U^\sigma_{k_1,...,k_d}\rightarrow\mathbb{R}$
by $G^{n,\sigma}_{k_1,...,k_d}(u)=\sum_{i=1}^{d+1}\lambda_i G_n(u_i)$
where $\lambda_1,...,\lambda_{d+1}\in [0,1]$ with $\sum_{i=1}^{d+1}\lambda_i=1$
and $\sum_{i=1}^{d+1}\lambda_i u_i=u$,
are uniquely determined.

Any element $u\in\mathbb{R}^d_{+}$ belongs to at least one simplex of the above form.
Observe that if $u$ belongs to two simplexes, say
$U^\sigma_{k_1,...,k_d}$ and $U^{\sigma'}_{k'_1,...,k'_d}$ then
$G^{n,\sigma}_{k_1,...,k_d}(u)=G^{n,\sigma'}_{k'_1,...,k'_d}(u).$
Thus we can extend the function $G_n:O_n\rightarrow \mathbb{R}$ to a function
$G_n:\mathbb{R}^d_{+}\rightarrow\mathbb{R}$ by setting
$G_n(u)=G^{n,\sigma}_{k_1,...,k_d}(u)$ for
$u\in U^\sigma_{k_1,...,k_d}$, where $k_1,...,k_d\in\mathbb{Z}_{+}$
and $\sigma:\{1,...,d\}\rightarrow\{1,...,d\}$ is a permutation.

This sequence has the desired properties.

\end{proof}

The following result
completes the proof of
theorem \ref{lem3.1}

\begin{lemma}  Suppose \reff{2.2} holds for all
bounded,  continuous functions.
Then, it also holds
for all bounded, upper semi-continuous $G$.
\end{lemma}
\begin{proof}
Let $G$ be bounded and upper semi-continuous.
Let $G_n$ be the sequence of bounded,
continuous functions constructed in the
previous lemma.
Hence \reff{2.2} and  Theorem
\ref{thm2.1} holds for $G_n$.

Using Theorem
\ref{thm2.1},
we choose
a sequence of  probability measures
$\mathbb{Q}_n\in\mathcal{M}_{\kappa,\cP}$
satisfying,
\begin{equation}
\label{app.2}
\mathbb{E}^{(n)} [G_n(\mathbb S)]>V(G_n)-\frac{1}{n}.
\end{equation}
Using similar compactness arguments
as in Lemma \ref{lem3.4},
we construct a
subsequence  $\mathbb Q_{n_l}$,
$l\in\mathbb N$ which converge weakly
 to a probability measure $\tilde{\mathbb Q}\in\mathcal{M}_{\kappa,\cP}$.
Recall that $G_n$'s are uniformly bounded.
Thus, by \reff{app.1} and the Skorohod representation theorem,
$$
\lim\sup_{l\rightarrow\infty}
\mathbb{E}^{(n_l)} [G_{n_l} (\mathbb S)]
\leq \tilde{\mathbb{E}} [G(\mathbb S)].
$$
This together with
\reff{app.2} yields that
$$
V(G)\leq \lim\inf_{n\rightarrow\infty}V(G_n)\leq
\lim\inf_{n\rightarrow\infty}\mathbb{E}^{(n)} [G_n(\mathbb S)]
\leq
\tilde{\mathbb{E}}[ G(\mathbb S)]\leq
\sup_{\mathbb Q\in \mathcal{M}_{\kappa,\cP} }\mathbb{E}_{\mathbb Q}
\left[G(\Sb)\right].
$$
This completes the proof.
\end{proof} 

\bibliographystyle{spbasic}

\end{document}